%% file: main.tex
\documentclass[12pt]{article}

\usepackage{authblk}
\usepackage[top=27mm, bottom=27mm, left=22mm, right=22mm]{geometry}
\usepackage{amsthm,amsmath,amssymb}
\usepackage{mathrsfs}
\usepackage{multirow}
\usepackage{amsthm}
\usepackage{microtype}
\usepackage{hyperref}

\usepackage{mathtools}
\usepackage{thm-restate}
\usepackage{thmtools}

\usepackage{tabularx}
\usepackage[table]{xcolor}

\usepackage[capitalize]{cleveref}
\usepackage{xcolor}

\newcommand{\chong}[1]{{\color{blue} {\small [Chong: #1]}}}

\newtheorem{theorem}{Theorem}[section]
\newtheorem{lemma}[theorem]{Lemma}
\newtheorem{proposition}[theorem]{Proposition}
\newtheorem{corollary}[theorem]{Corollary}

\newtheorem{definition}[theorem]{Definition}

\newtheorem{remark}[theorem]{Remark}
\newtheorem{claim}[theorem]{Claim}

\newtheorem{question}[theorem]{Question}

\DeclareMathOperator{\capa}{cap}

\begin{document}
\title{Singleton-type bounds for list-decoding and list-recovery, and related results}
\author[$*$]{Eitan Goldberg}
\author[$\dagger$]{Chong Shangguan}
\author[$*$]{Itzhak Tamo}
\affil[$*$]{\footnotesize Department of Electrical Engineering–Systems, Tel-Aviv University, Tel-Aviv 39040, Israel}
\affil[$\dagger$]{Research Center for Mathematics and Interdisciplinary Sciences, Shandong University, Qingdao 266237, China}
\affil[$*\dagger$]{\footnotesize Emails: eitang1@mail.tau.ac.il,  theoreming@163.com, zactamo@gmail.com}
\date{}
\maketitle

\begin{abstract}
List-decoding and list-recovery are important generalizations of unique decoding that received considerable attention over the years. However, the optimal trade-off among list-decoding (resp. list-recovery) radius, list size, and the code rate are not fully understood in both problems. This paper takes a step towards this direction when the list size is a given constant and the alphabet size is large (as a function of the code length). We prove a new Singleton-type upper bound for list-decodable codes, which improves upon the previously known bound by roughly a factor of $1/L$, where $L$ is the list size. We also prove a Singleton-type upper bound for list-recoverable codes, which is to the best of our knowledge, the first such bound for list-recovery. We apply these results to obtain new lower bounds that are optimal up to a multiplicative constant on the list size for list-decodable and list-recoverable codes with rates approaching  capacity.
%Our lower bounds are tight up to a constant factor, as shown by uniformly chosen random codes.

Moreover, we show that list-decodable \emph{nonlinear} codes can strictly outperform list-decodable linear codes. More precisely, we show that there is a gap for a wide range of parameters, which grows fast with  the alphabet size, between the size of the largest list-decodable nonlinear code and the size of the largest  list-decodable linear codes. This is achieved  by a novel connection between  list-decoding and the notion of  sparse hypergraphs in extremal combinatorics. We remark that such a gap is not known to  exist in the problem of unique decoding.

Lastly, we show that list-decodability or recoverability of codes implies in some sense good unique decodability.
%in different senses both for linear and general codes.
%although for some parameters better lower bounds for nonlinear codes are known. 

%thelinearor of codes in general and of Reed-Solomon (RS) codes in particular is not fully understood. We prove new relevant combinatorial bounds that generalize the Singelton bound for these problems and prove a difference in behavior for linear and general codes in achieving these bounds for list size of $2$. We also use the new size bounds to place an $\Omega_r(1/\epsilon)$ lower bound for the list size of codes of rate $1-r-\epsilon$ that are list-decodable or list-recoverable, finally we prove a connection between list-decoding or list-recovery and unique decoding of codes.
\end{abstract}

\section{Introduction}
As a generalization of unique decoding, the notion of list-decoding was introduced independently by Elias \cite{elias} and Wozencraft \cite{wozencraft} in the 1950s. In list-decoding, given a corrupted codeword, one can output a list of possible codewords, in contrast to unique decoding, where the output is at most one codeword. The advantage of list-decoding is that it can handle more adversarial errors than unique decoding. %Loosely speaking, a code of rate $R$ and length $n$ can correct at most $\frac{1-R}{2}n$ adversarial errors in unique decoding, while there exist codes that can correct at most $(1-R-o(1))n$ adversarial errors in list-decoding, where $o(1)$ tends to zero as $n$ tends to infinity.
List-recovery is a further generalization of list-decoding and was initially used as an intermediate step in the study of list-decoding and unique decoding (see \cite{guruswami2001expander,guruswami2002near,guruswami2003linear,guruswami2004efficiently} for example). Over the years, list-decoding and list-recovery have found many applications in information theory (see \cite{rudolf,Blinovski86,Blinovsky1997,Elias-survey91} for example) and theoretical computer science (see \cite{cai1999hardness,guruswami2009unbalanced,LP20,SIVAKUMAR1999270,STV01} for example).
%where a special interest is put on the list-decoding of Reed-Solomon codes (see \cite{cheng2007list,ferber2020list,guo2020improved,guruswami1998improved,guruswami2013linear,goldberg2021list,kopparty2018improved,shangguan2019combinatorial} for example).

Although extensively studied, many combinatorial properties of list-decoding and list-recovery are still far from being well-understood. In particular, the optimal trade-off between the radius of list-decoding/recovery, the list size, and the code rate are not known in both problems. In this paper, we take a step in this direction and consider this trade-off when the list size is constant and the alphabet size is large (as a function of the code length). 

To move forward let us introduce some notations and definitions. For a positive integer $q$, let $[q] =\{1,\ldots,q\}$. The {\it Hamming distance} $d(x,y)$ between two vectors $x,y\in [q]^n$ is the number of coordinates where they differ, namely, for $x,y\in[q]^n$, let  $d(x,y)=|\{i\in[n]:x_i\neq y_i\}|$. For an integer $1\le t\le n$ and a vector $v\in[q]^n$, let $B_t(v)$ denote the Hamming ball of radius $t$ centered at $v\in [q]^n$, which consists of all vectors in $[q]^n$ with Hamming distance at most $t$ from $v$.
A {\it code} $C$ of block length $n$ over an alphabet of size $q$ is a subset $C\subseteq [q]^n$, whose vectors are called codewords. The {\it rate} of $C$ is defined to be $R(C):=\log_q(|C|)/n$, and the {\it minimum distance} $d(C)$ of $C$ is defined to be the minimal Hamming distance between all pairs of distinct codewords in $C$, that is, $d(C)=\min\{d(x,y):x,y\in C,~x\neq y\}$. For simplicity, when $C$ is understood from the context, we will drop the dependencies of $C$ in $R(C)$ and $d(C)$ and just write $R$ and $d$. %We will also use $\delta:=d(C)/n$ to denote the {\it relative distance} of $C$. 

\paragraph{Unique decoding.} It is well-known and easy to see that for any code with minimum distance $d$, a Hamming ball of radius $\lfloor\frac{d-1}{2}\rfloor$ centered at any vector in $[q]^n$ can contain at most one codeword of the code. This implies that given a received codeword with at most $\lfloor\frac{d-1}{2}\rfloor$ corrupted coordinates, it is possible to decode it and output the correct (transmitted) codeword, simply by outputting the unique codeword in the Hamming ball of radius $\lfloor\frac{d-1}{2}\rfloor$ around it.
%if one is required to output a single codeword for unique decoding, then it suffices to find the codeword that is closest to this received vector. 
In the other direction, if one is requested to  correct any fraction $r$ of corrupted coordinates and output a unique codeword, the code must have a minimum distance of at least $2rn+1$.
%For codes with minimum distance $d$ the hamming balls of radius $\lfloor\frac{d-1}{2}\rfloor$ centered on any vector $v\in [q]^n$ have at most one codeword in them, thus $\lfloor\frac{d-1}{2}\rfloor$ errors in transmission can be corrected by a decoder.
The classical Singleton bound provides a bound on  the parameters of a code and shows that they must satisfy $d+Rn\leq n+1$. 
Then, it follows that for unique decoding, one must have $r\le (1-R)/2$ as $n$ tends to infinity.
Codes such as Reed-Solomon (RS) codes that attain this bound with equality are called MDS codes. 
%, and the equality can be achieved by Reed-Solomon codes.
%In other words, it indicates that in unique decoding one can successfully find the correct transmitted codeword if the fraction $r$ of corrupted coordinates satisfies $R\le 1-2r$. 
%Looking at this from a point of view focused on the fraction of errors that can be corrected, it states that if a fraction of $r$ errors can be fixed then 
%$r\le\frac{1-R}{2}$.

\paragraph{List-decoding and list-recovery.} Next, we will give the formal definitions of list-decoding and list-recovery, which are both a generalization of unique decoding.

\begin{definition}
A code $C\subseteq [q]^n$ is an {\it $(r,L)$ list-decodable} for $r\in (0,1),~L\in \mathbb{N}$ if $|B_{rn}(v)\cap C|\leq L$ for all $v\in [q]^n$, where $r$ and $L$ are called the {\it list-decoding radius} and the {\it list size} of the code, respectively.
\end{definition}

For positive integers $\ell$ ad $q$ let $\binom{[q]}{\leq \ell}$ be the family of subsets of $[q]$ of size at most $\ell$, i.e., $\binom{[q]}{\leq \ell}=\{S\subseteq [q]: |S|\leq \ell\}$. 
\begin{definition}
A code $C\subseteq [q]^n$ is an $(r,\ell,L)$ list-recoverable for $r\in (0,1)$, and $\ell,L\in \mathbb{N}$ if for any sequence of lists $S_1,\ldots,S_n\in \binom{[q]}{\leq \ell}$ the size of the set $\{c\in C:c_i\notin S_i\text{ for at most } rn \text{ coordinates}\}$ is at most $L$. Similarly, $r$ and $L$ are called the {\it list-recovery radius} and the {\it list size}, respectively.
\end{definition}

Clearly, one can easily see that unique decoding is a special case of list-decoding for $L=1$, and list-decoding is a special case of list-recovery for $\ell=1$. 

Given parameters $q$ and $r$, the list-decoding capacity $\capa_{LD}$ is the supremum over all rates  of  $q-$ary $(r,L)$ list-decodable codes  with $L$ that is at most polynomial in $n$. 
%all codes' rate so that for all $\epsilon>0$ and all sufficiently large $n>0$, there are $q$-ary codes of length $n$ and rate at least $\capa_{LD}-\epsilon$ that are $(r,L)$ list-decodable with $L\le poly(n)$. 
Similarly, given $q,r$ and $\ell$, the list-recovery capacity $\capa_{LR}$ is the 
supremum over all rates  of  $q-$ary $(r,\ell,L)$ list-recoverable codes with $L$ that is at most polynomial in $n$.
%largest rate so that for all $\epsilon>0$ and all sufficiently large $n>0$, there are $q$-ary codes of length $n$ and rate at least $\capa_{LR}-\epsilon$ that are $(r,\ell,L)$ list-recoverable with $L\le poly(n)$. 
By the list-decoding and list-recovery capacity theorems it is known  that $\capa_{LD}=1-h_q(r)$ for  $r\in[0,1-\frac{1}{q}]$ (see, e.g.,  \cite[Theorem 7.4.1]{guruswami2019essential}) and $\capa_{LR}=1-h_{q/\ell}(r)-\log_q(\ell)$ for  $r\in[0,1-\frac{\ell}{q}]$ (see, e.g.,  \cite{RudraW17}), where $h_q(x)$ is the $q$-ary entropy function defined in  \eqref{entropy}. 
Moreover, with high probability, a random code with rate $\capa_{LR}-\epsilon$ for $\epsilon>0$  is list-recoverable with list size $L=O(\ell/\epsilon)$, and  a similar result holds for list-decoding, simply by setting $\ell=1$.

Next, we describe our main results and compare them to the previously known results. 
%
%Our main results, as well as their comparison with previously 
%some known results, are summarized in the next subsection.
%\textcolor{green}{The very new paper \cite{roth2021higherorder} contains results overlapping some of our results, exactly which once will be specified.}
\subsection{Summary of main results}

%In this subsection we will summarize our main results. 
In the results reported below, we primarily  consider the  case of constant  list size $L$, independent of $n$. Moreover, similar to the classical Singleton bound, our bounds (results) behave well,  only when $q$ is sufficiently large as  a function of $n$. We begin with the first set of results that improve the Singleton-type bounds for list-decoding and recovery. 

\paragraph{Singleton-type upper bounds for list-decoding and list-recovery.} Shangguan and Tamo \cite{shangguan2019combinatorial} proved the following generalization of the Singleton bound to list-decoding. 

\begin{theorem}[Theorem 1.2 in \cite{shangguan2019combinatorial}] \label{thm:old-singelton}
For integers $q\ge 2,~L\ge 1$ and  $r\in [0,\frac{L}{L+1}]$ with $rn\in\mathbb{N}$, 
every $(r,L)$ list-decodable code $C\subseteq [q]^n$ has size at most $Lq^{n-\lfloor \frac{L+1}{L}rn\rfloor}$.
\end{theorem} 
It was also observed that one could prove a tighter upper bound for linear codes, as follows. %, as stated below.
\begin{proposition}[\cite{shangguan2019combinatorial}]\label{prop:old-singelton-linear} 
For integers $q>L\geq 1$, if $q$ is a prime power and $C\subseteq\mathbb{F}_q^n$ is a linear $(r,L)$  list-decodable code, then $|C|\le q^{n-\lfloor \frac{L+1}{L}rn\rfloor}$.
\end{proposition} 
\cite{shangguan2019combinatorial} also showed that the bound of Proposition \ref{prop:old-singelton-linear} is tight by showing that certain RS codes attain it with equality,  when $L=2,3$, $L\mid rn$, and $q$ is sufficiently large as a function of $n$. Hence, for this set of parameters 
%It follows that for general $r,L$ and $n$ 
the largest size of an  $(r,L)$ list-decodable code $C$ (not necessarily linear) satisfies $$q^{n-\lfloor\frac{L+1}{L}rn\rfloor}\le |C|\le Lq^{n-\lfloor\frac{L+1}{L}rn\rfloor},$$
and  the construction in \cite{shangguan2019combinatorial} is optimal up to a constant factor for $L=2,3$. We conjecture that the lower bound also holds for any $L, r$ and $n$ such  that $L\mid rn$, and that   it can be achieved by RS codes, see \cite{shangguan2019combinatorial}. 
%where we omitted the $\lfloor\cdot\rfloor$ symbols as we assumed that $L\mid rn$. 
Narrowing and even closing the gap between  the conjectured lower bound and the upper bound  is an interesting open question.
%, and also whether nonlinear codes improve on the current lower bound. 
%to know whether a much better construction exists if we allow the code to be nonlinear.

We improve the upper bounds in both \cref{thm:old-singelton} and \cref{prop:old-singelton-linear}, as detailed below.  

\begin{itemize}
\item %In \cref{subsec:gen-upp-bd} we show that it is unlikely to be true. 
In Section \ref{subsec:gen-upp-bd} we show that the factor $L$ in the upper bound of \cref{thm:old-singelton} can be replaced by $(1+o(1))$, where $o(1)$ tends to zero as $n$ tends to infinity. Thus, the new bound has the form  $(1+o(1))q^{n-\lfloor\frac{L+1}{L}rn\rfloor}$, which implies that 
%the is not tight, and can be replaced by $1+o(1)$ when $n$ is sufficiently large compared with $r$ and $L$ . This implies 
that the constructions given by \cite{shangguan2019combinatorial} is asymptotically optimal for $L=2,3$ and $L\mid rn$, even among nonlinear codes. See \cref{cor:singelton-lim-o(1)} for the formal statement.

\item  As already mentioned, Proposition \ref{prop:old-singelton-linear} is tight if $L$ divides $rn$ and $L=2,3$, and it is believed to be tight for any $L$ as long as  $L$ divides $rn$. In Section \ref{subsec:stronger-linear-bound}, we improve the upper bound exactly when this does not hold. More precisely,  we show that for sufficiently large $n$ compared to $r$ and $L$, every linear $(r,L)$ list-decodable code has dimension at most $n-\lceil\frac{L+1}{L}rn\rceil$ (see \cref{prop:linear-bound} for the formal statement). Hence, if $L\nmid rn$, the bound on the dimension is improved by one, compared to \cref{prop:old-singelton-linear}. This %relatively minor
improvement later enables us to separate linear and nonlinear codes by showing that there are nonlinear  $(r, L)$ list-decodable  codes whose size exceeds the size of any linear $(r, L)$ list-decodable  code.
%can have much more codewords than the largest linear codes can have in list-decoding.

\end{itemize}
It is known that $(r,L) $ list-decodable codes with alphabet size $q$ and rates approaching $1-h_q(r)(1+\frac{1}{L})$ exist (see Theorem 5.5 in \cite{guruswami2004list} and \cite{elias1991error}). This implies that for   alphabet of size at least  $q\ge 2^{\Omega(1/\epsilon)}$, there exist  $(r,L)$ list-decodable codes with rates approaching $1-\frac{L+1}{L}r$, as $h_q(r)$ approaches $r$ when $q$ is sufficiently large. This in turn implies that the Singleton-type bound for list decoding (Theorem \ref{thm:old-singelton}) and its improvements obtained in this paper provide an  asymptotically tight bound on the rate as $n$ tends to infinity. However, for a fixed $n$, it was unclear whether these Singleton-type bounds are tight. In Section \ref{sec:construcion} we show that this is indeed true,  by showing the existence of  \emph{nonlinear} codes with ``dimension'' $n-\frac{L+1}{L}rn-o(1)$, where the $o(1)$ term tends to zero for a fixed $n$ and $q$ tends to infinity % the above mentioned code of rate approaching $1-\frac{L+1}{L}r$ 
(see Section \ref{sec:construcion} for more details).

%on the rate does not imply by itself that the maximal dimension in Theorem \ref{thm:old-singelton} can  be achieved for any fixed $n$. In Section \ref{sec:construcion} we are able to prove the existence of codes with a dimension that is very close to the maxium given by Theorem \ref{thm:old-singelton} for any given $n$ and large enough $q$.

%\eitan{maybe this paragraph should be in a different part of the introduction (I comment here in that part). Originally I wrote it here as an argument that the bounds we introduced in this part of the introduction are actually good, maybe it fits more the part where we introduce our construction of bound achieving codes to show the difference between linear and nonlinear codes }

\begin{remark}
Recently, Roth \cite[Theorems 4,5]{roth2021higherorder} independently  proved the same  bound as in \cref{prop:linear-bound}, but under different assumptions. Furthermore, the bounds in \cite{roth2021higherorder} are stated as  bounds on the list-decoding radius; however, the bounds are equivalent.

\end{remark}

The last result in this set of results is  a new Singleton-type bound for $(r,\ell,L)$ list-recoverable codes, which  reduces to \cref{thm:old-singelton} for $\ell=1$. The reader is referred to \cref{sec:singleton-list-recovery} for details.

%\paragraph{Singleton-type upper bounds for list-recovery.}

\paragraph{Lower bounds on the list size.} The following is a typical question in the study of  list-decoding and list-recovery \cite{Blinovski86,blinovsky2005code,Guruswami-Average-Radius,Guruswami-random-linear,5605366} for examples). It is stated for list-recovery, and the corresponding question for list-decoding is obtained by setting $\ell=1$. . 

\begin{question}\label{que:list-size}
Given $q,r,\ell$ and  $\epsilon>0$ that measures the gap between the  code rate and the list-recovery capacity, what is the growth rate of the list size $L$?
%Prove that for sufficiently large $n$, any $(r,L)$ list-decodable (resp, $(r,\ell,L)$ list-recoverable) $q$-ary code of length $n$ and rate at least $\capa_{LD}-\epsilon$ (resp. $\capa_{LR}-\epsilon$) must have list size $L=\Omega(\frac{\ell}{\epsilon})$.
\end{question}
Previously, \cref{que:list-size} has been studied by several works (mainly for list-decoding), as discussed below.
%There are in fact quite a few works in the literature that are devoted to .
Blinovsky \cite{Blinovski86,blinovsky2005code} showed that any $(r,L)$ list-decodable code with rate $\capa_{LD}-\epsilon$ must have $L=\Omega(\log(\frac{1}{\epsilon}))$. Guruswami and Narayanan \cite{Guruswami-Average-Radius} studied that problem for average-radius list-decoding, which is a strengthening of list-decoding, and showed that the list size must be $\Omega(\frac{1}{\sqrt{\epsilon}})$.
Guruswami and Vadhan \cite{5605366} studied the regime of codes with list decoding radius of $r=(1-1/q)(1-\epsilon)$, approaching the upper limit of $1-1/q$, and proved that in  this regime  the list size must be $L=\Omega(1/\epsilon^2)$.
Lower bounds on the list size for list-decoding and list-recovery of \emph{random} codes were also studied in \cite{Guruswami-Average-Radius}  and \cite{Guruswami-random-linear}.  Guruswami and Narayanan \cite{Guruswami-Average-Radius}  proved that  both for random codes  and random linear codes of rate $cap_{LD}-\epsilon$ has list size $L=\Omega(1/\epsilon)$, where the hidden leading constant tends to zero as $r$ tends to $1-1/q$. Recently, Guruswami et al. \cite{Guruswami-random-linear} improved the leading constant  for  random \emph{linear} codes and showed that  $L\ge \lfloor h_q(r)/\epsilon+0.99\rfloor$. For the binary case they proved this lower bound is tight up to an additive constant, pinning down the list size for random binary linear codes to a range of three values. Additionally for list recovery  \cite{Guruswami-random-linear} showed that for a random linear $(0,\ell,L)$ list-recoverable code of rate $1-log_q(\ell)-\epsilon$, it holds that $L=\ell^{\Omega(1/\epsilon)}$.

It is known that if one allows the alphabet size $q$ to grow, 
then there exist random codes with rate $1-r-\epsilon$ that are list-decodable  (resp. list-recoverable) from radius  $r$  and list size $L=O(1/\epsilon)$ (resp. $L=O(\ell/\epsilon)$). In fact,  it is sufficient that  $q\ge 2^{\Omega(1/\epsilon)}$. Hence, in this case the two capacities coincides, and we have $\capa_{LD}=\capa_{LR}=1-r$. In Proposition \ref{prop:list-size-LR} we partially answer \cref{que:list-size} by showing that  
an $(r,\ell,L)$ list-recoverable code  of length $n$ and rate at least $1-r-\epsilon$ must satisfy $L\ge \frac{\ell r}{\epsilon}+\ell-1+o(1)=\Omega_r(\ell/\epsilon)$, where $o(1)$ tends to zero as $n$ tends to infinity. The special case of list decoding is proved  in Proposition \ref{prop:list-size-LD}.

%, given constants $r$ and $\ell$, for all $\epsilon>0$ and all sufficiently large $n$, there exist codes, e.g., completely random codes, 
%Therefore one has $\capa_{LD}=\capa_{LR}=1-r$ if $q$ is allowed to grow exponentially with $\epsilon$. In this case, we answer \cref{que:list-size} in the affirmative
%, and show that
%\begin{itemize}
%    \item 
%\end{itemize}
%by showing that any $(r,L)$ list-decodable code (resp. $(r,\ell,L)$ list-recoverable) of length $n$ and rate at least $1-r-\epsilon$ must satisfy $L\ge \frac{r}{\epsilon}+o(1)$ (resp. $L\ge \frac{\ell r}{\epsilon}+\ell-1+o(1)$), where $o(1)$ tends to zero as $n$ tends to infinity (see Propositions \ref{prop:list-size-LD} and  for formal statements). 

%\begin{restatable}[Goldbach's conjecture]{theorem}{goldbach}
%\label{thm:goldbach}
%Every even integer greater than 2 can be expressed as the sum of two primes.
%\end{restatable}

% \cref{thm:goldbach}:

%\goldbach*

%The results listed above are in fact direct consequences of the Singleton-type upper bounds for list-decoding and list-recovery, respectively.  

%Theorems \ref{thm:old-singelton} and \ref{thm:recovery-singelton}, respectively.

\paragraph{Nonlinear codes outperform linear codes in list-decoding.} A fundamental problem in combinatorial coding theory is to obtain an optimal trade-off between the rate of a code and its relative distance $\delta$. However, this problem is far  from being solved despite decades of research. Moreover, even the more modest problem of understanding the power of nonlinear codes is unknown. More precisely,  it is unknown whether nonlinear codes perform better than linear codes under unique decoding or linear codes perform as well as their nonlinear counterpart. Indeed, 
%under the that whether linear codes have best performance in this problem. 
for binary codes with relative distance  $0<\delta<1/2$, the currently best known lower and upper bounds are the GV  \cite{gilbert1952comparison,varshamov1957estimate} and MRRW \cite{mceliece1977new} bounds, respectively. It is well known that linear code can achieve the GV  bound, and there is no better upper bound for linear codes that is tighter than the MRRW bound. Although some stronger lower bounds are known for nonlinear codes (see \cite{jiang2004asymptotic,vu2005improving}), it is not known whether there is a gap between the size of the largest linear code and nonlinear code for a given distance.
 
Surprisingly, considering the current  state of the art for unique decoding, in Section \ref{sec:construcion} we show that such a separation between linear and nonlinear codes exists for  list-decoding. In particular,  we  show  that nonlinear codes can considerably outperform linear codes for sufficiently large $q$ as a function of $n, L$.
%the cardinality of the largest generic list-decodable code is considerably larger than the cardinality of the largest linear list-decodable code. 
Roughly speaking, we show that for given $r,L$ with $rn\in\mathbb{N}$ and $L\nmid rn$, there is a constant $\theta\in[\frac{1}{L},1)$ so that the size of the largest linear  $(r,L)$ list-decodable code is at most a $q^{-\theta}$-fraction of the size of the largest nonlinear  $(r,L)$ list-decodable code. In particular, for  $L=2$, one can take $\theta=1-\epsilon$ with any  $\epsilon>0$ arbitrarily close to zero, provided that $q$ is sufficiently large as a function of $n,\epsilon$. The precise statement of this result can be found in \cref{prop:hypergraphstocodes} and \cref{rem:comparison}. 
We provide new constructions of list-decodable codes via a correspondence between codes and multi-partite hypergraphs to derive this result. In particular, the constructions are based on a notion of sparse-hypergraphs from extremal combinatorics. We use several known constructions of such sparse hypergraphs in the literature to construct the codes.
%The details can be found in \cref{sec:construcion}. 

We note that results of similar flavor, i.e., that nonlinear codes perform better than linear codes, are known to exist; however, they are scarce. In particular, for the problem of \emph{erasure list-decoding}, it is known that there exist nonlinear codes whose  list size is exponentially smaller than the list size guaranteed for linear code (see  \cite[Theorem 10.17]{guruswami2004list}).
Another example of this phenomenon is the recent result by \cite{Guruswami-random-linear} that showed that  in the problem of zero-error list-recovery,  random codes also  have significantly smaller list sizes than  random linear codes.

%decoding Also a similar difference about the list size of linear and general codes was proven for the related problem of \emph{Erasure list-decoding} 

\paragraph{Large list-decoding radius implies large minimum Hamming distance.}
We study the relation between  list-decodability and unique decodability of a code. In particular, whether a code with good list-decoding properties necessarily imply unique decoding properties. We divide the analysis into two cases, depending on whether the code is linear or nonlinear (see Section \ref{sec:distances-mds}).
%of the code influences its unique decodability, we part the discussion to general and to linear codes, both disscused in .
\begin{itemize}
    \item 
For a general list-decodable code, i.e., not necessarily linear,  we prove that it  must contain a large subcode with a large minimum Hamming distance (see Theorem \ref{thm:generl-lower-bound-for-distance}). As a corollary of this theorem, we obtain that a list-decodable code with a rate approaching the maximal rate given by Theorem \ref{thm:old-singelton} (and whose existence is guaranteed by  \cref{prop:hypergraphstocodes}) must contain a very large near MDS subcode (see Corollary\ref{cor:general-bound-achivers-are-MDS}).
\item  For a linear code that is list-decodable or  recoverable, we show that it must have a large Hamming distance. This result can be viewed as a  generalization of  the fact that an $(r,1)$ list-decodable code (uniquely  decodable) has Hamming distance  of at least $2rn+1$,   to list-decoding and recovery. For details, see Theorem \ref{thm:linear-distance-recovery} and its derivatives Proposition \ref{prop:linear-distance-decoding} and Corollary \ref{cor:linear-bound-achiver-is-MDS}.
\end{itemize}

As a final remark, we note that \cite{roth2021higherorder} independently proved a result regarding the unique decodability of list-decodable linear codes (see  \cite[Theorem 3]{roth2021higherorder}), which  is equivalent to Corollary \ref{cor:linear-bound-achiver-is-MDS} when $L$ divides $rn$. Moreover, it can be verified  from the proof of \cite[Theorem 3]{roth2021higherorder} that Roth also proved  Proposition \ref{prop:linear-distance-decoding}. The proof argument of Roth is  very similar to ours,  though it is  formulated differently.

\subsection{Organization} The rest of this paper is organized as follows. In \cref{sec:prelim} we introduce some necessary notations and definitions. In Sections \ref{sec:combi-bounds} and \ref{sec:singleton-list-recovery} we present the Singleton-type upper bounds for list-decoding and list-recovery, respectively. In \cref{sec:construcion} we introduce the notion of sparse hypergraphs and use them to show that for a wide range of parameters, the largest generic list-decodable codes must have much more codewords than the largest linear list-decodable codes. In \cref{sec:distances-mds} we show that if a linear code has a very large list-decoding or list-recovery radius, then it must also have a very large minimum Hamming distance. 
 
\section{Preliminaries and notations}
\label{sec:prelim}
We will use of the following notations. For positive integers $m\leq n$, we write $[n]=\{1,\ldots,n\}$, $[m,n]=\{m,\ldots,n\}$, $\binom{[n]}{m}=\{A\subseteq [n]:|A|=m\}$, and $\binom{[n]}{\le m}=\{A\subseteq [n]:|A|\le m\}$. We number vectors by  superscripts, i.e., $x^1,x^2,\ldots$, and  use  subscripts to refer to their  coordinates, e.g.,  $x^j_i$ is the $i$th coordinate of $x^j$. 
For a subset $I\subseteq [n]$ and a vector $x$ of length $n$, let $x_I$ be the restriction of $x$ to its coordinates with indices in $I$. For $x,y\in [q]^n$ let $I(x,y)=\{i:x_i=y_i\}$ be the set of indices for which $x$ and $y$ are equal, then it is clear that $d(x,y)+|I(x,y)|=n$, where  $d(x,y)$ is the hamming distance between $x$ and $y$. We will use $h_q(x)$ to denote the $q$-ary entropy,
\begin{equation}
    \label{entropy}
h_q(x) := x\log_q(q - 1) - x\log_q(x) - (1 - x)\log_q(1 - x). 
\end{equation}

For $n$ subsets $S_1,\ldots,S_n\subseteq[q]$, let $S_1\times\cdots\times S_n$ be the set of vectors $v\in[q]^n$ with $v_i\in S_i$ for all $i\in[n]$. For a set of vectors $D\subseteq [q]^n$ and a vector $v\in [q]^n$, let $d(v,D):= \min\{d(v,u):u\in D\}$. Using the above notation, it is not hard to check by definition that a code $C$ is $(r,\ell,L)$ list-recovery if and only if for every $D\in \binom {[q]}{\leq \ell}^n$, $$|\{c\in C:d(c,D)\leq rn \}|\leq L$$ where we define $\binom {[q]}{\leq \ell}^n=\{S_1\times\cdots\times S_n:S_i\in\binom{[q]}{\le \ell} \text{~for all $i\in[n]$}\}$. 

%\chong{Add several sentences here.}

%\chong{I changed the notation to $\binom {[q]^n}{\leq l}$, but I am not sure if I understood you correctly. Did you use the notion $d(c,D)$ in the proof?}

%When $q$ is a prime power, $[q]$ can be viewed as the final field $\mathbb{F}_q$ of $q$ elements. In this case 
For a prime power $q$, let $\mathbb{F}_q$ be the finite field of $q$ elements. A code $C\subseteq \mathbb{F}_q^n$ is linear if and only if it is a subspace of $\mathbb{F}_q^n$. RS codes \cite{RS-codes} is an important family of linear codes, which is  %is the family of Reed-Solomon codes \cite{RS-codes} (RS-codes for short), that is 
defined as follows: for integers $k\leq n\leq q$ and a vector $\alpha \in \mathbb{F}_q^n$ with distinct entries, the RS code with evaluation vector $\alpha$ is the $k$-dimensional subspace $$\{(f(\alpha_1),\ldots,f(\alpha_n):f\in \mathbb{F}_q[x],~deg(f)\leq k-1\}.$$

Since for every $r'\leq r$ an $(r,L)$ list-decodable code is also $(r',L)$ list-decodable, by Theorem \ref{thm:old-singelton} we see that an $(r,L)$ list-decodable code with $r\ge \frac{L}{L+1}$ has size at most $L$, which is obviously tight. So in order to avoid trivial cases, throughout the paper we assume that $r\in[0,\frac{L}{L+1})$.

%\chong{I think it is better to organize the paper as follows. What do you think?}

\section{Singleton-type bounds for list-decoding}\label{sec:combi-bounds}

\subsection{Upper bounds for arbitrary codes}\label{subsec:gen-upp-bd}

%Given list-decoding radius $r$ and list size $L$, in this subsection we will present several new upper bounds on the size of a generic $(r,L)$ list-decodable code.

\iffalse
Before stating our result, let us first recall the following Singleton-type bound for list-decoding, proved by Shangguan and Tamo \cite{shangguan2019combinatorial}.

\begin{theorem}[Theorem 1.2 in \cite{shangguan2019combinatorial}] \label{thm:old-singelton}
For integers $q\ge 2,~L\ge 1$ and a real $r\in [0,\frac{L}{L+1}]$ with $rn\in\mathbb{N}$, 
every $(r,L)$ list-decodable code $C\subseteq [q]^n$ has size at most $Lq^{n-\lfloor \frac{L+1}{L}rn\rfloor}$.
\end{theorem}
\fi

%Since for every $r'\leq r$ an $(r,L)$ list-decodable code is also $(r',L)$ list-decodable, by Theorem \ref{thm:old-singelton} we see that an $(r,L)$ list-decodable code with $r\ge \frac{L}{L+1}$ has size at most $L$, which is obviously tight. So in order to avoid trivial cases, throughout the paper we assume that $r\in[0,\frac{L}{L+1})$.

Below we will present several upper bounds that improve upon \cref{thm:old-singelton}. We begin with the following theorem, which provides an  improved upper bound on the cardinality of any list-decodable codes.

%Our first bound shows that for sufficiently large $n$ (compared with $r,~L$), the factor $L$ in the upper bound in Theorem \ref{thm:old-singelton} can be in fact replaced by $1+o(1)$, where $o(1)$ tends to zero as $n$ tends to infinity (see \cref{thm:new-o(1)-singelton} and \cref{cor:singelton-lim-o(1)}). 

%We can even get rid of the $o(1)$ term when $n$ is considerably larger (compared with $r, L$ and $q$, see \cref{cor:singelton-ngreaterq} for the details).
%\textcolor{red}{
%\begin{theorem}\label{thm:new-o(1)-singelton2}
%For integers $q\ge 2,~L\ge 1$ and a real $r\in [0,\frac{L}{L+1})$ with $rn\in\mathbb{N}$, there exists an integer $n(r,L)$ such that for all $n\geq n(r,L)$, every $(r,L)$ list-decodable code $C\subseteq [q]^n$ satisfies
%$$|C|\leq \max\{ q+\lfloor f(n)q\rfloor,L\}\cdot q^{n-(\lfloor\frac{L+1}{L}rn \rfloor+1)},$$ 
%where $f(n)=\frac{(L-b-1)}{2(\lfloor\frac{L+1}{L} rn \rfloor+1)-(L-b-1)}$ and $b\in[0,L-1]$ satisfies $b\equiv rn \pmod{L}$.
%\end{theorem}}

\begin{theorem}\label{thm:new-o(1)-singelton}
Let $C\subseteq [q]^n$ be an $(r,L)$ list-decodable code with  $r\in [0,\frac{L}{L+1})$ and $rn\in\mathbb{N}$, then for large enough $n$ (as a function of  $r$ and $L$) the size of $C$ satisfies
$$|C|\leq \max\{ q(1+ o_{r,L}(1)),L\}\cdot q^{n-(\lfloor\frac{L+1}{L}rn \rfloor+1)},$$ 
where $o_{r,L}(1)$ is a function that tends to zero for fixed $r,L$ and  $n\rightarrow\infty$. % tends to infinity.
\end{theorem}
Note that by inspecting the precise function $o_{r, L}(1)$ to be given in the proof of the theorem, one can verify  that  \cref{thm:new-o(1)-singelton} reduces to the classical Singleton bound for $L=1$. Moreover, the theorem holds for any $n\geq L^2/r$. 

%Note that the $f(n)$ is not a function of $q$ thus the asymptotic behavior does not change also for cases where $q$ changes with $n$, and that the bound from Theorem \ref{thm:new-o(1)-singelton} can be viewed as a generalization of the well-known Singelton bound, as it boils down to it since $f(n)=0$ for $L=1$. Moreover, for our proof it suffices to take $N(r,L)=\frac{L^2}{r}$.

\begin{proof}
We will need the following observation. Let $C\subseteq [q]^n$ be an $(r,L)$ list-decodable code, then  for  $I\subseteq [n]$ and $w\in [q]^{|I|}$ the set of  vectors $\{c_{\overline{I}}:c\in C,~c_I=w\}\subseteq [q]^{|\overline{I}|}$ is  $(\frac{rn}{n-|I|},L)$ list-decodable.

 Let $m:=\Big\lfloor\frac{L+1}{L}rn \Big\rfloor+1$ and note that  $m\leq n$, since  $r<\frac{L}{L+1}$,  and let $b\in[0,L-1]$   such that  $b\equiv rn \pmod{L}.$
%If $C\subseteq [q]^n$ is an $(r,L)$-list-decodable code, then for any subset of coordinates $I\subseteq [n]$ and any vector $w\in [q]^{|I|}$ the set of vectors $\{c_{\overline{I}}:c\in C,~c_I=w\}\subseteq [q]^{|\overline{I}|}$ is an $(\frac{rn}{n-|I|},L)$-list-decodable. Indeed, otherwise the code $C$ would not be $(r,L)$-list-decodable. 
%Let $m:=\lfloor\frac{L+1}{L}rn \rfloor+1$ and note that $m\leq n$ since  $r<\frac{L}{L+1}$. Let also  $R$ to be the unique integer in $[0,L-1]$ such that  $R\equiv rn \pmod{L}.$
We claim that for any vector $w\in [q]^{n-m}$ there are less than   $M:=1+\max\{ q+\lfloor f(n)q\rfloor,L\}$, 
where $f(n)=\frac{(L-b-1)}{2(\lfloor\frac{L+1}{L} rn \rfloor+1)-(L-b-1)}$ codewords $c\in C$ with $c_{[m+1,n]}=w$, then the result will follow since    %we have 
\begin{align*}
|C|=\sum_{w\in [q]^{n-m}}|\{c\in C:c_{[m+1,n]}=w\}|&\leq \max\{ q+\lfloor f(n)q\rfloor,L\}\cdot q^{n-m}\\
&=\max\{q(1+o_{r,L}(1)),L\}\cdot q^{n-(\lfloor\frac{L+1}{L}rn \rfloor+1)}.
\end{align*}
 
%\begin{align*}
%  |C|&=\sum_{w\in [q]^{n-m}}|\{c\in %C:c_{[m+1,n]}=w\}|\leq Mq^{n-m}\\
%  &=\max\{q+\lfloor f(n)q\rfloor,L\}\cdot q^{n-(\lfloor\frac{L+1}{L}rn \rfloor+1)}.
%\end{align*}
 %
%And the result follows since   $f(n)$ tends to zero for fixed $L,r$ and $n\rightarrow\infty$. therefore, it remains to prove the claim. 
Assume towards a contradiction that the claim is false, then $C$ contains   $M$ codewords whose last $n-m$ coordinates are all identical. Let  $v^1,\ldots, v^{M}\in [q]^m$ be the restriction of these codewords to their first $m$ coordinates. The contradiction will follow by showing that  the set of vectors $\{v^i:~i\in[M]\}\subseteq [q]^m$ is \emph{not} $(\frac{rn}{m},L)$ list-decodable, together with  the observation above.

Towards this end, let us construct a multi-graph whose vertices are the vectors $v^i,~i\in[M]$, and draw an edge between distinct $v^i$ and $v^j$ for every coordinate they agree on. 
%It follows easily from the pigeonhole principle \chong{(Do we have to mention pigeonhole principle here? Can we just say that ``It is not hard to see...'')} 
It is not hard to verify that each coordinate $i\in[m]$ contributes at least $M-q\geq 0$ edges to the multi-graph. 
%\chong{The proof here is correct. But can we do better than $M-q$? Say for $j=1,\ldots,q$ let $x_j$ be the number of $j$'s in the $i$-th coordinate, then we have $\sum_{j=1}^q x_j=M$. The number of edges given by the $i$-th coordinate is  $\sum_{j=1}^q\binom{x_j}{2}$, which by the convexity of $\binom{x}{2}$ is at least $\frac{M(M-q)}{2q}$. Does this help? (I think maybe not)}
Therefore, the multi-graph has average degree at least $\frac{2(M-q)m}{M}>L-b-1$, where the inequality follows since $M>q+f(n)q$. As the degree of a vertex must be an integer, there exists a vertex (vector) $v\in\{v^i:~i\in[M]\}$ of degree at least $L-b$. Equivalently, % the following inequality holds 
\begin{equation}
\label{stam}
    \sum_{u\in U}|I(u,v)|\ge L-b,
\end{equation}  
 where  $U=\{v^i:~i\in[M]\}\setminus\{v\}$.  By the value of $M$, $|U|\geq L$, then it is possible to remove vertices from $U$ to make it the set of vertices $\{v^i:i\in [L]\}$ (possibly by changing the indices of the vertices), while still maintaining \eqref{stam}. Suppose next, that $n$ is sufficiently large so that $\lfloor\frac{rn}{L}\rfloor+1\ge L$, then there exists a subset of coordinates  $A\subseteq [m]$ be  of size $\lfloor\frac{rn}{L}\rfloor+1$ with $\sum_{u\in U} |A\cap I(v,u)|\ge L-b$.

Partition the set  $[m]\backslash A$ arbitrarily to $L$ sets $P^i,i=1,\ldots,L$, each of size  at least $\lfloor \frac{rn}{L}\rfloor+1-|A\cap I(v^i,v)|$. This is  possible since 
\begin{align*}
\sum_{i=1}^L\Big\lfloor \frac{rn}{L}\Big\rfloor+1-|A\cap I(v^i,v)|\leq
L\Big(\Big\lfloor \frac{rn}{L}\Big\rfloor+1\Big)- (L-b)=rn=m-|A|.
\end{align*}
To complete the proof it suffices to construct a vector $y\in [q]^m$ such that the ball  $B_{rn}(y)$ contains the vectors $v,v^1,\ldots,v^L$.  
 %at least $L+1$ distinct vectors  $v^i$. 
 %Let $w^j$  for $j=1,\ldots, L-|U'|$ be distinct vectors in the set $\{v^i:i\in [M]\}\backslash (\{v\} \cup U')$ (this is well-defined as $M\ge L+1$).
 %, this set is of size at least $L-|U'|$ and therefore such vectors $w^j$ exist.  
 Define the vector $y$ by
$$
y_A=v_A, \text{ and }~y_{P^i}={v}^i_{P^i} \text{ for } i\in [L].
$$
It is clear that $d(y,v)\le m-|A|=rn$.
Furthermore,  since $y$ and $v$ agree on the coordinates in $A$, then $v^i$ agrees with $y$ on $|I(v,v^i)\cap A|$ coordinates in $A$ (which is possibly zero), and by construction on at least $\lfloor \frac{rn}{L}\rfloor+1-|I(v,v^i)\cap A|$ coordinates in $[m] \backslash A$. Therefore 
$$d(y,v^i)\le m-|I(v,v^i)\cap A|-\Big(\Big\lfloor \frac{rn}{L}\Big\rfloor+1-|I(v,v^i)\cap A|\Big)=rn,$$
and we  have arrived at the desired contradiction which completes the proof of the theorem.
\end{proof}

%Below we list some consequences of Theorem \ref{thm:new-o(1)-singelton}. 
%has some interesting consequences, as listed below.

The following corollary, which is an easy consequence of \cref{thm:new-o(1)-singelton}, shows that for sufficiently large $n$ (as a function of $r,~L$), the factor $L$ in the upper bound in  \cref{thm:old-singelton} can be replaced by $1+o(1)$, where $o(1)$ tends to zero as $n$ tends to infinity.

\begin{corollary}
\label{cor:singelton-lim-o(1)}
For integers $q\ge 2,~1\le L\le q$ and  $r\in [0,\frac{L}{L+1})$ with $rn\in\mathbb{N}$, if $n$ is sufficiently large with respect to $r,L$,
%there exists an  $N(r,L)$ such that for all $n\geq N(r,L)$, 
then every $(r,L)$ list-decodable code in $ [q]^n$ has size at most 
$$(1+o(1))q^{n-\lfloor\frac{L+1}{L}rn \rfloor},$$
where $o(1)$ tends to zero as $n$ tends to infinity.
\end{corollary}
\begin{proof}
Apply Theorem  \ref{thm:new-o(1)-singelton} and note  that for $L\leq q$, one has $\max\{q+\lfloor(f(n)q\rfloor,L\}\leq (1+(f(n))q$ and   $f(n)=o(1)$. 
\end{proof}
Corollary \ref{cor:singelton-lim-o(1)}, which is  an improvement over  Theorem \ref{thm:old-singelton},  is of interest since it provides an asymptotically optimal bound on the size of such codes. Indeed,  the last two authors showed in \cite{shangguan2019combinatorial}  that over sufficiently large finite fields,  $L=2,3$ and $ L\mid rn$, there exist RS codes of size $q^{n-\frac{L+1}{L}rn}$. 
%Therefore, \cref{cor:singelton-lim-o(1)} shows that these RS codes given in \cite{shangguan2019combinatorial} have not only optimal rate (according to \cref{prop:old-singelton-linear}) but also near optimal size. 

With an additional condition, one can remove the $o(1)$ term in the statement of  \cref{cor:singelton-lim-o(1)}, and obtain a cleaner bound, as follows.  
\iffalse
The following two results show that with some additional conditions, we can remove the $o(1)$ term in the statement of  \cref{cor:singelton-lim-o(1)}.

\chong{Our Singleton-type only performs well when $q$ is large compared with $n$. So I want to delete \cref{cor:singelton-ngreaterq} as it only holds for $n\ge n(q,r,L)$.}  

\begin{proposition}\label{cor:singelton-ngreaterq}
For integers $q\ge 2,~1\le L\le q$ and a real $r\in [0,\frac{L}{L+1})$ with $rn\in\mathbb{N}$, there exists an integer $n(q,r,L)$ such that for all
$n\geq n(q,r,L)$, every $(r,L)$ list-decodable code in $ [q]^n$ has size at most   
$$ q^{n-\lfloor\frac{L+1}{L}rn \rfloor}.$$
\end{proposition}

\begin{proof}
Apply Theorem \ref{thm:new-o(1)-singelton} for any sufficiently large $n$ that satisfies $qf(n)<1$. %This is possible since $f(n)$ tends to zero and is constant in $q$.
\end{proof}

Corollary \ref{cor:singelton-ngreaterq}holds only for  large enough $n$ compared with $r, L$ and $q$, which does not necessarily stand for codes whose alphabet sizes are increasing with the block length (e.g., RS codes). \cref{cor:singelton-remainder} below shows that we may also get rid of the dependency on $q$. 
\fi

\begin{corollary}
\label{cor:singelton-remainder}
For integers $q\ge 2,~1\le L\le q$ and  $r\in [0,\frac{L}{L+1})$ with $rn\in\mathbb{N},~rn\equiv L-1 \pmod{L}$, 
if $n$ is sufficiently large with respect to $r,L$,
%there exists an  $N(r,L)$ such that for all $n\geq N(r,L)$, 
then every $(r,L)$ list-decodable code in $ [q]^n$ has size at most
%there exists an  $n(r,L)$ such that for all  $n\geq n(r,L)$, every $(r,L)$ list-decodable code in $ [q]^n$ satisfies
$q^{n-\lfloor\frac{L+1}{L}rn \rfloor}.$
\end{corollary}

\begin{proof}
Apply Theorem \ref{thm:new-o(1)-singelton} and note that $f(n)=0$ if $rn\equiv L-1 \pmod{L}$.
\end{proof}

The next result, which can be deduced from either \cref{thm:old-singelton} or \cref{thm:new-o(1)-singelton}, gives a lower bound on the list size of list-decodable codes. Moreover, it  partially answers  \cref{que:list-size}. 

\begin{proposition}\label{prop:list-size-LD}
   Any $q$-ary $(r,L)$ list-decodable code of length $n$ and rate at least $1-r-\epsilon$ satisfies $L\ge \frac{r}{\epsilon}+o(1)$, where $o(1)$ tends to zero as $n$ tends to infinity.
\end{proposition}

\begin{proof}
Let $C$ be a code  that satisfies the assumption of the proposition, then by Theorem \ref{thm:old-singelton} 
$$q^{(1-r-\epsilon)n}= |C|\leq Lq^{n-\lfloor\frac{L+1}{L}rn\rfloor}.$$
Equivalently, $(1-r-\epsilon)n\leq n-\lfloor\frac{L+1}{L}rn\rfloor +\log_qL$ and  the result follows by rearranging.
%one gets  $\frac{rn}{L}\le\epsilon n+\log_q L+\frac{L+1}{L}+1$. 
%As $\frac{\log_q L+\frac{L+1}{L}+1}{n}$ tends to zero as $n$ tends to infinity, we have that  $L\ge\frac{r}{\epsilon}+o(1)$.
\end{proof}

\subsection{Improved upper bounds for linear codes} \label{subsec:stronger-linear-bound}
This section shows that the upper bounds obtained in Section \ref{subsec:gen-upp-bd} can be further improved when restricted to linear codes. We begin with the following lemma.
\begin{lemma}\label{lemma:strong-linear-bound}
For a prime power $q$, positive integers $n,L\ge 2$, and $r\in[0,\frac{L}{L+1})$ satisfying $rn\in\mathbb{N}$, and $n-\lceil \frac{L+1}{L}rn\rceil+1>(L-1)\frac{q}{q-1}$, any $[n,n-\lceil\frac{L+1}{L}rn\rceil+1]_q$ is \emph{not}   $(r,L)$ list-decodable.
\end{lemma}
As $(L-1)\frac{q}{q-1}$ is decreasing with $q$, the lemma in fact holds for all large enough $n$ satisfying $n+1-\lceil \frac{L+1}{L}rn\rceil>2(L-1)$, where we set $q=2$.
In  \cite{roth2021higherorder} Roth gives examples of two $[n,n-\lceil\frac{L+1}{L}rn\rceil+1]_q$ codes that  \emph{are}   $(r,L)$ list-decodable, seemingly contradicting Lemma \ref{lemma:strong-linear-bound}. The codes are the $[n,n-1]_q$ parity code, which is  $(1/n,n)$ list-decodable and its dual, the $[n,1]_q$ repetition code for $n=(L+1)u-1$ for some $u,L\in \mathbb{N}$, which is $(\frac{Ln-1}{(L+1)n},L)$ list-decodable. One can easily verify that the parameters of these two codes do not satisfy the assumptions of Lemma \ref{lemma:strong-linear-bound}.
\begin{proof}
Since $rn\in\mathbb{N}$ we can write $rn=La+b$ for integers $a,b$ with $b\in [0,L-1]$. 
The following can be easily verified 
$$\Big\lceil\frac{L+1}{L}rn\Big\rceil-1 =
\left\{
	\begin{array}{ll}
		(L+1)a+b = rn +a & \mbox{if } L \nmid rn  \\
		(L+1)(a-1)+L=rn +a-1 & \mbox{if } L\mid rn
	\end{array}
\right.
$$

Let $C\subseteq\mathbb{F}_q^n$ be an $[n,k]$-linear code with $k=n-(\lceil\frac{L+1}{L}rn\rceil-1)$. As before, to prove the lemma it suffices to show that there exist $L+1$ distinct codewords that are contained in a ball of radius $rn$.
%$c^1,c^2,...,c^{L+1}\in C$ and a vector $y\in\mathbb{F}_q^n$ such that $\{c^1,c^2,...,c^{L+1}\}\subseteq B_{rn}(y)$. 

Towards this end, assume without loss of generality that  the first $k$ coordinates of the code form  an information set, and consider the $k(q-1)$ codewords of the code whose restriction to this information set is a vector of weight one, i.e., it has only one nonzero coordinate.
%iBy the property of an information set for each coordinate $j\in[k]$ there are $q-1$ distinct codewords $c\in C\backslash \{c^1\}$ with $c_{[k]\backslash \{j\}}=0^{k-1}$. Again, by the property of an information set it is easy to verify that the $k(q-1)$ codewords given by the previous step are all distinct.
As $k(q-1)>(L-1)q$, by the pigeonhole principle,  among these $k(q-1)$ codewords there are $L$ codewords, say $c^1,...,c^{L}\in C$, agree on their $k+1$ coordinate, i.e., $c^i_{k+1}=c^j_{k+1}$ for  $i,j\in [L]$.

We will  consider two cases, $L\nmid rn$ and $ L\mid rn$ and  notice   that $|[k+2,n]|=n-k-1=\lceil\frac{L+1}{L}rn\rceil-2$. In the first case  $n-k-1=(L+1)a+b-1\ge (L+1)a$, hence there is a partition of $[k+2,n]$ into $L+1$ pairwise disjoint subsets, say $P_1,\ldots,P_{L+1}$, with $|P_j|\ge a$ for each $j\in [L+1]$. 
Next, let $y\in\mathbb{F}_q^n$ be the vector  satisfying 
$$y_{[k]}=0,~y_{k+1}=c^2_{k+1}, y_{P_j}=c^j_{P_j}~\text{for each $j\in[L]$}, \text{~and~} y_{P_{L+1}}=0.$$
It is routine to check that  $d(y,0),d(y,c^j)\le rn$ for each $j\in [L+1],$ and we have obtained the desired contradiction. 

Similarly, if $L|rn $ then  $|[k+2,n]|=\lceil\frac{L+1}{L}rn\rceil-2=(L+1)(a-1)+L-2$,
hence there is a partition of $[k+2,n]$ into $L+1$ pairwise disjoint subsets, with $|P_j|\ge a-1$ for each $j\in [L+1]$. 
Let $y\in\mathbb{F}_q^n$ be the vector satisfying
$$y_{[k]}=0,~y_{k+1}=c^2_{k+1}, y_{P_j}=c^j_{P_j}~\text{for each $j\in[L]$}, \text{~and~} .$$
As before  $d(y,0),d(y,c^j)\le rn$ for each $j\in [L+1],$ which contradicts the assumptions of list-decodability.
\end{proof}

The reformulation of Lemma \ref{lemma:strong-linear-bound} gives the following proposition.

\begin{proposition}\label{prop:linear-bound}
For a prime power $q$, an integer $2\le L$, and  $r\in [0,\frac{L}{L+1})$ with $rn\in\mathbb{N}$, there exists an integer $n(r,L)$ such that for all $n\ge n(r,L)$ any $[n,k]_q$ code that is $(r,L)$ list-decodable satisfies  $k\leq n-\lceil\frac{L+1}{L}rn\rceil$. 
\end{proposition}
During the work on this paper, we became aware of a recent paper by Roth \cite{roth2021higherorder}  who proved  a result similar to Proposition  \ref{prop:linear-bound} in Theorems 4,5 of \cite{roth2021higherorder}. These theorems provide the same bound on the dimension of the code as  Proposition \ref{prop:linear-bound} does; however, they assume a bit stronger assumptions. More precisely,    
\begin{itemize}
    \item Theorem 4 assumes that the code is MDS with rate greater than $1-\frac{2}{L}-\frac{(n-k)\mod (L+1)}{n}.$ 
    \item Theorem 5 assumes that  the code is MDS, alphabet $q>\binom{n}{k+1}$ and list size $n-k-1\leq L<\binom{n}{k}.$
    %\item \textcolor{green}{Theorem 6: The inequality is proved for average radius list deocdable linear codes.}
\end{itemize}
Proposition \ref{prop:linear-bound}  almost subsumes Theorem 4, except for a small number of cases for which   Theorem 4  holds and Proposition\ref{prop:linear-bound} does not hold.
However, it does not subsume Theorem 5 since Proposition \ref{prop:linear-bound} assumes that  $n$ is large enough compared to $L$, so it does not hold for $L$ too large, for example, when $L\approx \binom{n}{k}$.

\begin{proof}
This is just the contrapositive of Lemma \ref{lemma:strong-linear-bound}, since linear codes have integer dimensions.
\end{proof}

We note that the method of \cite{shangguan2019combinatorial} also gives the following result, whose proof is omitted.
 
\begin{proposition}\label{prop:RS-achive-linear}
 For any sufficiently large $q$ and any real $r\in[0,\frac{2}{3})$ with $2\nmid rn$ there exist $[n,n-\frac{3rn+1}{2}]$-RS codes that are also $(r,2)$ list-decodable.
\end{proposition}

Proposition \ref{prop:linear-bound} implies that any linear code that satisfies the parameters of Proposition \ref{prop:RS-achive-linear} has dimension at most $n-\lceil\frac{3rn}{2}\rceil=n-\frac{3rn+1}{2}$, and hence shows that the construction given by Proposition \ref{prop:RS-achive-linear} is also optimal among all linear codes (in the corresponding parameter regime).

\section{Nonlinear codes outperform linear codes in list-decoding} \label{sec:construcion}
In this section, we show that there exist \emph{nonlinear} codes whose list-decodability outperform any other linear code, i.e., they strictly outperform their linear counterpart. Our method will exploit some known results in the area of extremal (hyper)-graph theory and equivalence between certain  ``sparse hypergraphs'' and codes with ``good'' list-decodability properties. We begin first by the equivalence.
%more preciselyin the problem of combinatorial list-decoding, one can provide examples which indicate that nonlinear codes can {\it strictly} outperform linear codes. To prove our result we will make use of a connection between codes and hypergraphs, in particular a notion known as ``sparse hypergraphs''.

\subsection{An equivalence  between codes and multi-partite hypergraphs}

%Below we will introduce a correspondence between codes and multi-part hypergraphs. 

Let us begin with some needed definitions. A {\it hypergraph} $H$ is an ordered pair $H=(V,E)$, where the {\it vertex set} $V$ is a finite set and the {\it edge set} $E$ is a family of distinct subsets of $V$. A hypergraph is called {\it $n$-uniform} if all of its edges are of size $n$. An $n$-uniform hypergraph is further called {\it $n$-partite} if its vertex set $V$ admits a partition $V=\cup_{i=1}^n V_i$ such that every edge intersects each vertex set $V_i$ in {\it exactly} one vertex. %For clarity throughout we will not distinguish between a hypergraph $H$ and its edge set $E(H)$.

We will define a natural bijection between  the family of $n$-uniform $n$-partite hypergraphs with equal part size $q$ (i.e., $|V_i|=q$ for all $i\in [n]$) and the family of $q$-ary codes of length $n$, as follows. For  $i\in [n]$, let $V_i=\{(i,a):a\in[q]\}$ and $V=\cup_{i=1}^n V_i$. For an $n$-uniform $n$-partite hypergraph $H=(V,E)$ (with equal part size $q$) we associate  a code of size  $|E|$ in the following way: for each edge $e=\{(i,x_i):i\in [n], x_i\in [q]\}\in E$ define the codeword $$\psi(e):=(x_1,\ldots,x_n)\in [q]^n,$$ and the code $$C_{H}:=\{\psi(e):e\in E\}\subseteq [q]^n.$$
Clearly, the mapping $\psi$ is a bijection, 
%and onto $[q]^n$ \chong{Why `onto $[q]^n$`? Do we really need this property?}, 
and for a vector $x\in[q]^n$  define $\psi^{-1}(x)=\{(i,x_i):i\in [n], x_i\in [q]\}$. Then, one can easily verify  that the Hamming distance between any two vectors $x,y\in [q]^n$   satisfies 
\begin{align}\label{eq:needed}
  |\psi^{-1}(x)\cap \psi^{-1}(y)|+d(x,y)=n.
\end{align} 
Next, we will introduce the notion of sparse hypergraphs, which will later be used to construct nonlinear list-decodable codes, using the mapping $\psi$ defined above. 

For positive integers $v\ge 2,~e\ge 2$, an $n$-uniform hypergraph $H$ is called $(v,e)$-sparse if for any $e$ distinct edges $A_1,\ldots,A_e\in H$, it holds that $|\cup_{i=1}^ e A_i|> v$. Let $g_n(q,v,e)$ denote the maximum number of edges in any $(v,e)$-sparse $n$-uniform $n$-partite hypergraph with equal part size $q$. In what follows, we will list some known lower bounds on $g_n(q,v,e)$.

\begin{lemma}[Section 4, Brown, Erd\H{o}s, S\'os \cite{bes1}]\label{lem:BES}
Given positive integers $q,n,e$, there is a  $\mu_1>0$ depending only on $n,e$ such that $g_n(q,v,e)\ge\mu_1 q^{\frac{en-v}{e-1}}$.  
\end{lemma}

%\textcolor{red}{let's give the precise number of the theorems in the papers.}
\begin{lemma}[Theorem 3, Shangguan, Tamo \cite{shangguan2020sparse}]\label{lem:ST}
Given positive integers $q,n,e$ with $\gcd(en-v,e-1)=1$, there is  $\mu_2>0$ depending only on $n,e$ such that $g_n(q,v,e)\ge\mu_2 q^{\frac{en-v}{e-1}}\cdot\log^{\frac{1}{e-1}} q$.
\end{lemma}

\begin{lemma} [Theorem 1, Alon, Shapira \cite{t=3}]\label{lem:Alon-Shapira}
Given  integers $s,n$ with $2\le s< n$ and  $\epsilon>0$,
%For every real $\epsilon>0$ and all positive integers $s,n$ with $2\le s\le n-1$, 
there is an integer $q(n,\epsilon)$ such that for all $q\geq q(n,\epsilon)$ it holds that 
%there exists an $n$-uniform $n$-partite hypergraph with equal part size $q$ which is $(3n-2s+1,3)$-sparse and has at least $q^{s-\epsilon}$ edges.
$g_n(q,3n-2s+1,3)\geq q^{s-\epsilon}$.
\end{lemma}

%The following well-known result about existence of sparse hypergraphs with many edges is due to Alon and Shapira \cite{t=3}. We remark that although not explicitly stated, the hypergraph constructed in the proof of Theorem 1 in \cite{t=3} is in fact a  3-partite hypergraph with equal part size. 

We remark that although all of the lower bounds in Lemmas \ref{lem:BES}, \ref{lem:ST} and \ref{lem:Alon-Shapira} were initially proved for arbitrary $n$-uniform hypergraph, which is not necessarily $n$-partite, one can easily convert them to lower bounds on $n$-uniform $n$-partite hypergraphs by the following variant of the Erd\H{o}s-Kleitman lemma (see Theorem 1, \cite{Erdos-Kleitman}). 

%, which states that any $n$-uniform hypergraph $H$ with $m$ edges has an $n$-uniform $n$-partite subhypergraph with at least $\frac{n!}{n^n}\cdot m$ edges. 
%\textcolor{red}{(Chong, we need to be more precise, since I don't think their lemma  guarantees the parts to be of equal size. We need to modify our claim a bit. Also, since the proof of the lemma is an application of the first moment method, I suggest to state that proof in two-three lines, what do you think?)}

\begin{lemma}
Any $n$-uniform hypergraph $H$ with $m$ edges has an $n$-uniform $n$-partite subhypergraph with equal part size and at least $\frac{n!}{n^n}\cdot m$ edges.
\end{lemma}

\begin{proof}
Let $H=(V,E)$. To prove the lemma, we can assume without loss of generality that $n\mid |V|$, as if $n\nmid |V|$, by adding to $H$ $n- |V|\pmod{n}$ isolated vertices, the following proof still works. Take a uniformly chosen random partition of $V$, such that every subset in the partition has equal size $q:=\frac{|V|}{n}$. Let $F$ be the $n$-uniform $n$-partite hypergraph given by such a random partition. It is clear that $F$ has equal part size $q$, and moreover, every $n$-subset of $V$ lies in $F$ with equal probability $\frac{q^n}{\binom{nq}{n}}>\frac{n!}{n^n}$. Therefore, by the linearity of expectation, the expected number of edges contained in $F$ is at least $\frac{n!}{n^n}\cdot m$, as needed.
\end{proof}

\subsection{Constructions of list-decodable codes via sparse hypergraphs}

The following lemma shows that one can construct ``good" list-decodable codes from sparse hypergraphs.

\begin{lemma}\label{lem:hypergraphstocodes}
If $H=(V,E)$ is an $(n+(L+1)rn,L+1)$-sparse $n$-uniform $n$-partite hypergraph, then the code $C_H=\{\psi(e):e\in E\}\subseteq [q]^n$ is  $(r,L)$ list-decodable. 
\end{lemma}

\begin{proof}
Suppose for the sake of contradiction that $C_H$ is not $(r,L)$ list-decodable. Then there exist $L+1$ distinct codewords $c^1,\ldots,c^{L+1}\in C_H$ and a vector $y\in [q]^n$ such that $d(c^i,y)\le rn$ for  $1\le i\le L+1$. It therefore follows that $|\psi^{-1}(c^i)\setminus\psi^{-1}(y)|\le rn$, and moreover
\begin{equation*}
    \begin{aligned}
    |\bigcup_{i=1}^{L+1} \psi^{-1}(c^i)|&\le&|\psi^{-1}(y)|+\sum_{i=1}^{L+1}|\psi^{-1}(c^i)\setminus\psi^{-1}(y)|
    &\le&n+(L+1)rn,
    \end{aligned}
\end{equation*}
which contradicts the assumption that $H$ is $(n+(L+1)rn,L+1)$-sparse.
\end{proof}  

Applying \cref{lem:hypergraphstocodes} in concert with the lower bounds of $g_n(q,v,e)$ listed in Lemmas \ref{lem:BES}, \ref{lem:ST}, \ref{lem:Alon-Shapira} and \ref{lem:hypergraphstocodes} gives the main result of this section, whose proof is omitted as it  follows straightforwardly. 

\begin{proposition}\label{prop:hypergraphstocodes}
\begin{enumerate}
    \item For integers $q\ge 2, L\ge 1$ and a real $r\in[0,\frac{L}{L+1})$ with $rn\in\mathbb{N}$ there exists an $(r,L)$ list-decodable code $C\subseteq [q]^n$ with $|C|\geq \mu_1 q^{n-\frac{rn(L+1)}{L}}$, where $\mu_1>0$   depends only on $n,L$;
   
    \item For integers $q\ge 2, L\ge 1$ and a real $r\in[0,\frac{L}{L+1})$ with $rn\in\mathbb{N}$ and  $\gcd(L,rn)=1$, there exists an $(r,L)$ list-decodable code $C\subseteq [q]^n$ with $|C|\geq \mu_2 q^{n-\frac{rn(L+1)}{L}}\cdot\log^{\frac{1}{L}} q$, where $\mu_2>0$   depends only on $n,L$;
    
    \item For positive integer $n$ and a real $r\in[0,\frac{2}{3})$ with $rn\in\mathbb{N}$ odd and  $\epsilon>0$, there is an integer $q(n,\epsilon)$ such that for all $q\geq q(n,\epsilon)$ there exist  an $(r,2)$ list-decodable code $C\subseteq [q]^n$ with $|C|>q^{n-\frac{3rn-1}{2}-\epsilon}.$
\end{enumerate} 
\end{proposition}

%Our main result is the following theorem.

%\begin{theorem}\label{thm:first-statement}
%  For every real $\epsilon>0$ and   positive integers $m,n$ with $n> 3m+1$, there is an integer $q(n,\epsilon)$ such that for all $q\geq q(n,\epsilon)$ there exist  a non-linear $(\fr{2m+1}{n},2)$ list-decodable code in $[q]^n$ of size at least $q^{n-3m-1-\epsilon}.$
%\end{theorem}
\begin{remark} \label{rem:comparison}
\cref{prop:hypergraphstocodes} shows the existence of nonlinear codes that are better list-decodable than all other linear codes with the same parameters. Indeed, it shows that for $L\ge 2,~rn$ not divisible by $L$ and sufficiently large $q$ (as a function of $n,L$), there are $(r,L)$ list-decodable codes of size $\Omega_{n,L}(q^{n-\frac{rn(L+1)}{L}})$. However, \cref{prop:linear-bound} shows that for  sufficiently large $n\ge n(r,L)$, every linear $(r,L)$ list-decodable code over  $\mathbb{F}_q$ has dimension at most $n-\lceil\frac{L+1}{L}rn\rceil$, equivalently it is of size at most 
$ q^{n-\lceil\frac{L+1}{L}rn\rceil}$. Hence, the upper bound given in \cref{prop:linear-bound} does  not hold in general. 

We find this phenomenon quite surprising; although other instances of nonlinear codes outperform all other linear codes are known in coding theory, they are few and far between. 
%recalling Recall that . \cref{prop:hypergraphstocodes} indicates that if
%$L\nmid rn$ then for sufficiently large $q$ (as a function of $n,L$), the upper bound given by \cref{prop:linear-bound} is not tight for generic codes. Therefore, in the problem of combinatorial list-decoding, there exist nonlinear codes which strictly outperform linear codes, which is quite surprising compared with the situation in the problem of unique decoding. 
%\cref{thm:first-statement} is interesting in the sense that Proposition \ref{prop:linear-bound} indicates that in the same parameter regime, a $(\fr{2m+1}{n},2)$ list-decodable linear code has size at most $q^{n-3m-2}$, which is roughly a $\frac{1}{q^{1-o(1)}}$-fraction of the size of the code given by Theorem \ref{thm:first-statement}.
\end{remark}

\section{Singleton-type bound for list-recovery}\label{sec:singleton-list-recovery}
In this section, we prove the following Singleton-type bound for list-recovery.
\begin{theorem}
\label{thm:recovery-singelton}
For integers $q\ge 2,~\ell\ge 1,L\ge \ell$ and  $r\in[0,1-\frac{\ell}{L+1})$ with $rn\in \mathbb{N}$, every $(r,\ell,L)$ list-recoverable code $C\subseteq [q]^n$ satisfies
$$|C|\leq L{\Big(\frac{q}{\ell}\Big)}^{n-\lfloor\frac{L+1}{L+1-\ell}rn\rfloor}.$$
\end{theorem}

%Note that for $r=0$ \cref{thm:recovery-singelton} recovers the $1-log_q(l)$ rate bound in the list-recovery capacity theorem, and 
Note that when $\ell=1$ \cref{thm:recovery-singelton} recovers Theorem \ref{thm:old-singelton}, hence it is can be seen as  a generalization of the Singleton bound.  \textcolor{black}{Additionally, this bound recovers the list-recovery capacity bound on the  rate   $R\leq 1-log_q(\ell)$,  for zero error ( i.e., $r=0)$ and  constant list size $L$.}
%\chong{Can we say anything about other parameters for comparison with existing bounds?}

\begin{proof}
%Consider an $(r,\ell,L)$ list-recoverable  code  $C\subseteq [q]^n$ and  set 
Let $t:=\lfloor\frac{\ell}{L+1-\ell}rn\rfloor$ and $m:=rn+t$. The proof of the theorem  will  follow from the  following two claims.

%We will make use of the next claim.
\begin{claim}
\label{claim:L-is-max}
For every $B\in \binom{[q]}{\le \ell}^{n-m}$ there are at most $L$ codewords $c\in C$ with $c_{[m+1,n]}\in B$. 
%\chong{I can understand what you want to say. But did you define the meaning of $c_{[m+1,n]}\in B$? (Here $c_{[m+1,n]}$ is a vector but $B$ seems to be a product of some subsets.)}
\end{claim}

\begin{proof}
Recall that for a positive integer $n$, $\binom {[q]}{\leq \ell}^{n}=\{S_1\times\cdots\times S_n:S_i\in\binom{[q]}{\le \ell} \text{~for all $i\in[n]$}\}$, where $\binom{[q]}{\le \ell}$ is the family of subsets of $[q]$ of size at most $\ell$.
Suppose towards a contradiction that there exist $B\in \binom{[q]}{\le \ell}^{n-m}$ and $L+1$ codewords $c^1,\ldots,c^{L+1}\in C$ such that  $c^i_{[m+1,n]}\in B$, for all $i\in [L+1]$. 
%Assume without loss of generality that $B=S_{m+1}\times\cdots\times S_{n}$, where $S_i\in\binom{[q]}{\le \ell}$ for all $i\in[m+1,n]$. 
To obtain the desired contradiction, it is enough to show that there exist sets $S_1,\ldots,S_m\in\binom{[q]}{\le \ell}$ such that $d(c^i,S_{1}\times\cdots\times S_{m}\times B)\le rn$ for all $i\in [L+1]$.  
Since  $m=rn+t$ and $c^i_{[m+1,n]}\in B$, it suffices to construct sets $S_i$ such that $d(c^i_{[m]},S_{1}\times\cdots\times S_{m})\le rn$. 
In other words, for each $i$ there exists at least $t$ indices $j\in [m]$ such that $c^i_j\in S_j$. The sets $S_j$ be can be easily constructed since  the number of constraints they have to satisfy is $(L+1)t$, whereas there are $m$ sets $S_i$, each of size $\ell$, and the choice of the parameters $(L+1)t\leq \ell m$. Table \ref{tab:litreview} shows an example of a  selection of such sets for certain values of the parameters. 
%However, this is vacuously true, as $\lceil\frac{(L+1)t}{m}\rceil\le\ell$ by our choice of parameters, we will have ``enough space'' to place the subsets $S_1,\ldots,S_m$ so that each $c^i_{[m]}$ has at least $t$ coordinates choosing from those them. 
\input{setconstrustion}

%Define the next function $$f:[m]\times [\ell]\to \{c^1,...,c^{L+1}\}$$ to be an assignment of codewords to coordinates such that every codeword is assigned at least $t$ times consecutively, by moving from left to right from bottom to top (first coordinate changes cyclically and second grows for each cycle of the first).     
%This is possible since
%\begin{align*}
%    &\frac{\ell m}{L+1}=\frac{\ell(rn+t)}{L+1}\ge t
%\end{align*}

%Define for every $i\in[m]$, $S_i=\{f(i,j)_i:j\in[\ell]\}$, which are lists of size at most $\ell$. Notice also that since $\frac{lm}{L+1}<m$ that for any $j_1\ne j_2$ and $i$ $f(i,j_1)\ne f(i,j_2)$ thus $c^j_i\in S_i$ for at least $t$ coordinates $i$ for every $j\in[L+1]$.
%Now define  $$D=\prod\limits_{i\in [m]}S_i\times B.$$ It is easy to check that $D\in\binom{[q]}{\le\ell}^n$, and moreover for every $j\in [L+1]$, 
%$c^j$ has the property that 
%$d(c^j,D)\leq n-(n-m)-t=rn$, contradicting the assumption that $C$ is $(r,\ell,L)$-list-recoverable. 

\end{proof}
 
%To prove the theorem, it is enough to prove one more claim.
 
 \begin{claim}
 \label{claim:recursive-size-in-box}
 For any $t\leq n-m$, and any $S_1,\ldots,S_{n-m-t}\in \binom{[q]}{\leq \ell}$, there are at most $L(\frac{q}{\ell})^t$ codewords $c\in C$ with $c_i\in S_i$.
 \end{claim}
 %We prove this claim using induction on $t$ with the goal of bounding the code's size taking $t=n-m$. We use Claim \ref{claim:L-is-max} as the base case of the induction.
 \begin{proof}
 Let us apply induction on $t$. For the base case  $t=0$, the statement follows from \cref{claim:L-is-max}.  
 %induction base, here $t=0$ and this is reduced to Claim \ref{claim:L-is-max}.
 Now in order to prove the claim for $t\le n-m$, let us assume that we have proved it for $t-1$. Fix $S_1,\ldots,S_{n-m-t}\in \binom{[q]}{\leq \ell}$.
 %Now in the induction step let $t\ge 1$ and we asume that we have proved the claim for $t-1$. 
 %Let any $S_1,...,S_{n-m-t}\in \binom{[q]}{\leq l}$. 
 For $j\in [q]$ let $a_j$ be the number of codewords $c\in C$ with $c_i\in S_i$ for all $i\in [n-m-t]$ and $c_{n-m-t+1}=j$. It is easy to see that $$\{c\in C:c_i\in S_i \text{ for all $i\in[n-m-t]$}\}=\sum_{j\in[q]} a_j,$$ so to prove the claim it suffices to show $\sum_{j\in [q]} a_j\le L(\frac{q}{\ell})^t$. 
 %we are interested in an upper bound for the amount $\sum_{j\in [q]}a_j$.
By induction hypothesis for $t-1$, for any $A\in\binom{[q]}{\ell}$, $\sum_{j\in A}a_j\leq L(\frac{q}{\ell})^{t-1}$. Therefore, by averaging over all such sets  $A\in\binom{[q]}{\ell}$ it can be easily seen that $\sum_{j\in [q]} a_j\le L(\frac{q}{\ell})^t$, completing the proof of the claim.
%For this optimization problem the maximum is easily seen to be achieved when $a_j=\frac{L(\frac{q}{l})^{t-1}}{\ell}$ for all $j$. Thus the amount of such code words is at most $q\frac{L(\frac{q}{l})^{t-1}}{\ell}=L(\frac{q}{l})^t$.
 \end{proof}

Returning to the proof of \cref{thm:recovery-singelton}, one can see that it  follows directly from   \cref{claim:recursive-size-in-box} with $t=n-m$.
%we apply Claim \ref{claim:recursive-size-in-box} with $t=n-m$ placing no restrictions on values of the codewords in any of the coordinates, and thus we find that $|C|\leq L(\frac{q}{l})^{n-m}=L{(\frac{q}{\ell})}^{n-\lfloor\frac{L+1}{L+1-\ell}rn\rfloor}$ as needed.
\end{proof}

The following corollary, which can be deduced easily from \cref{thm:recovery-singelton}, gives a lower bound on the list size of list-recoverable codes. It also provides a partial answer to \cref{que:list-size}, as explained in the introduction.

\begin{corollary}
 \label{prop:list-size-LR}
 Any $q$-ary $(r,\ell,L)$ list-recoverable code of length $n$ and rate at least $1-r-\epsilon$ satisfies $L\ge \frac{\ell r}{\epsilon}+\ell-1+o(1)$, where $o(1)$ tends to zero as $n$ tends to infinity.
\end{corollary}

\begin{proof}
By Theorem \ref{thm:recovery-singelton} the size of any such code  $C$ satisfies 
$$q^{(1-r-\epsilon)n}\leq |C|\leq L{\Big(\frac{q}{\ell}\Big)}^{n-\lfloor\frac{(L+1) rn}{L+1-\ell} \rfloor}\leq L{q}^{n-\lfloor\frac{(L+1) rn}{L+1-\ell} \rfloor}.$$
Hence, $\frac{\ell}{L+1-\ell}rn  \leq \epsilon n +log_q(L)+1+\frac{L+1}{L+1-\ell}$, and  as $\Big (log_q(L)+1+\frac{L+1}{L+1-\ell}\Big)\Big/n$ %$\frac{log_q(L)+1+\frac{L+1}{L+1-\ell}}{n}$ 
tends to zero as $n$ tends to infinity, the result follows.
\end{proof}

\section{Large list-decoding radius implies  a large minimum distance} \label{sec:distances-mds}

In this section, we show that good list-decoding property would imply in some sense good unique decoding property, i.e., large minimum distance. We divide our analysis into two cases; first, for general codes (not necessarily linear),  we show that such a statement can not hold in general, but it is undoubtedly true for a large subcode of the  code. Then, we proceed to consider the case of linear codes, where we are able to prove that even for the general problem of list-recovery, the code itself (and not its subcode) must have a large minimum distance. 
%\zac{we need to say something about Rony's work}\eitan{I wrote about it in the introduction}
%of    the if a linear code has large list-decoding (list-recovery) radius then it must also have large minimum Hamming distance. Our main result is the following theorem.

\vspace{0.2cm}
{\bf General codes:} It is clear that for general codes, one can \emph{not} hope that  good list-decodability would imply a large minimum distance. Indeed, given an $(r, L-1)$ list-decodable code, one can ``ruin'' the minimum distance by adding a new codeword that is of distance one from one of its codewords (and thereby possibly making it a nonlinear code). The new code is $(r, L)$ list-decodable code and thus retains its good list-decoding property; however, it has a poor minimum distance.  On the other hand, one needs only to remove the newly  added codeword to obtain back the (possibly) large minimum distance of the code. In other words, a small number of codewords needs to be removed to have also a large minimum distance (and, of course, retain the good list-decoding property). The following theorem shows that this is the only case in general, i.e., any large enough code with good list-decoding property must contain a large subcode  with a large minimum distance. The formal details follow.

\begin{theorem}
\label{thm:generl-lower-bound-for-distance}
Let  $L\ge 1$ be an integer,  $\gamma \in (0,1),~r\in [0,\frac{L}{L+1})$,
 $n\in \mathbb{N}$ with  $rn\in\mathbb{N}$, 
 and $q\ge q(n,r,L,\gamma)$, then 
 every $(r,L)$ list-decodable code $C\subseteq [q]^n$ with $|C|=q^{n-\lfloor\frac{L+1}{L}rn \rfloor-\epsilon }$, where $\epsilon:=\epsilon(n)\geq 0$ is an integer valued function  such that $(L-1)(\epsilon+1) \leq \lfloor rn/L\rfloor$, contains a subcode of size at least  $ \gamma |C|$ and minimum distance at least  $\lfloor\frac{L+1}{L}rn \rfloor-(L-1)(\epsilon+1)+1$.
\end{theorem}
%\textcolor{green}{
We note that in the more general case where $\epsilon$ is not necessarily an integer, a slightly weaker result holds, where one can show that the subcode distance is weakened to be at least $\lfloor\frac{L+1}{L}rn \rfloor-(L-1)(\lfloor\epsilon\rfloor+2)+1$, if $(L-1)(\lfloor\epsilon\rfloor+2) \leq \lfloor rn/L\rfloor$. We omit the details. 
%proof as it is not much rewarding and more technical.}

Applying Theorem \ref{thm:generl-lower-bound-for-distance} to a  sequence of codes with $\epsilon(n)=o(n)$ and a fixed $L$ gives the following corollary.
%\textcolor{green}{
\begin{corollary}
\label{cor:general-bound-achivers-are-MDS}
Let $L\ge 1$ be an integer,  $\gamma \in (0,1),~r\in [0,\frac{L}{L+1})$,
 $n\in \mathbb{N}$ with  $rn\in\mathbb{N}$, 
 and $q\ge q(n,r,L,\gamma)$, then every $(r,L)$ list-decodable code $C\subseteq [q]^n$ with $|C|=q^{n-\lfloor\frac{L+1}{L}rn \rfloor-o(n) }$ contains a subcode of size at least $\gamma |C|$ and minimum distance at least  $\lfloor\frac{L+1}{L}rn \rfloor-o(n)$.
\end{corollary}
%\begin{proof}
%Applying Theorem \ref{thm:generl-lower-bound-for-distance} results in the corollary as for a $\epsilon= o(n)$ function $(L-1)(\epsilon+2)+1$ is also $o(n)$.
%\end{proof}}
Notice that the guaranteed subcode by Corollary \ref{cor:general-bound-achivers-are-MDS} has near optimal rate-distance tradeoff. Indeed, the rate  is at least 
$ 1-\delta +\frac{1}{n}-o(1)$, whereas the rate of an MDS code is $1-\delta +\frac{1}{n}$. 
%$1-\delta +\frac{1}{n}-\frac{o(n)-o(n)-1+log(\gamma)}{n}= 1-\delta +\frac{1}{n}-o(1)$, whereas the rate of an MDS code is $1-\delta +\frac{1}{n}$. 
Moreover,  Proposition \ref{prop:hypergraphstocodes} shows the existence of such sequence of codes with $\epsilon<1$ and  large enough $q$.

We proceed with the proof Theorem \ref{thm:generl-lower-bound-for-distance}.
%, which  is in the same spirit of the other proofs of bounds in this paper. 

\begin{proof}
Let $C\subseteq [q]^n$ be a code that satisfies the assumptions of the theorem. The required subcode  will be constructed by identifying and removing a small subset of codewords from the code that cause its minimum distance to be relatively small. Therefore, the remained codewords, i.e., the subcode, will  have a large minimum distance. 
%result in a near MDS code.
We proceed with the formal proof. 
%describe a way of removing codewords from $C$ which will result in a near MDS code. We will also bound the amount of codewords we remove to complete the proof.  

Let $m:=\lfloor\frac{L+1}{L}rn \rfloor+\epsilon+ 1$, %be defined as in the proof of Theorem \ref{thm:new-o(1)-singelton}, 
and  say that a vector $w\in [q]^{n-m}$ is bad for the subset $I\subseteq [n],|I|=n-m$ if there exist two codewords $c^1,c^2\in C$ such that $c^1_{I}=c^2_{I}=w$, and $c^1_j=c^2_j$ for at least $L(\epsilon+1)$ additional coordinates $j\notin I$, and note that by the assumption on $\epsilon$,  $L(\epsilon+1)\leq m$. Next, we get a  bound on  the number of codewords whose projection on a fixed set $I$ is a bad vector for $I$. 
\begin{claim}
\label{claim:bad-is-small}
If $w\in[q]^{n-m}$ is bad for the set $I$, then there are at most $L$ codewords $c\in C$ with $c_{I}=w$.
\end{claim}

Now we can proceed to construct the desired subcode $C'$ with large size and distance, as follows.  $C'$ is obtained from $C$ by removing from it all codewords $c$ such that there exists a subset $I\subseteq [n]$ of size $n-m$, and a vector $w\in q^{n-m}$ that is bad for $I$ and $c_I=w$. 
%for every subset $I\subseteq [n]$ of size $n-m$, and for every vector $w\in q^{n-m}$ that is bad for itrelative to this set $I$ (replace $[m+1,n]$ for $I$ in the definition of a bad vector), remove all codewords $c\in C$ such that $c_I=w$
%remove f 
%to the step of removing codewords with the goal of gaining a subcode with the goal distance and size. First note that the choice of the set $[m+1,n]$ in Claims \ref{claim:bad-is-small}, \ref{claim:most-w-have-L+1} was arbitrary and that they are true for every subset of $[n]$ of size $n-m$.
%We remove codewords by the next rule, for every subset $I\subseteq [n]$ of size $n-m$, and for every vector $w\in q^{n-m}$ that is bad relative to this set $I$ (replace $[m+1,n]$ for $I$ in the definition of a bad vector), remove all codewords $c\in C$ such that $c_I=w$. Let $C'\subseteq C$ be the remaining subcode.
We claim that 
$C'$ has distance at least $\lfloor\frac{L+1}{L}rn \rfloor-(L-1)(\epsilon+1)+1$, and size at least $\gamma |C|$.

 \vspace{0.2cm}
 {\bf Size:}
 As there are $\binom{n}{n-m}$ sets $I\subseteq [n]$ of size $n-m$, and for each such set $I$, there are at most $q^{n-m}$ bad vectors $w$ for it. By Claim \ref{claim:bad-is-small} there are at most $L$ codewords $c\in C$ such that $c_I=w$.  Hence  at most $$\binom{n}{n-m}q^{n-m}L,$$
 codewords were removed  and the size of $C'$ is a least  
 $$ q^{n-m+1}-L\binom{n}{n-m}q^{n-m}=\Big(1-\frac{L\binom{n}{n-m}}{q}\Big)|C|\ge \gamma |C|,$$
% \begin{align*}
%      &q^{1-\epsilon}q^{n-m}-L\binom{n}{n-m}(1-\frac{q^{1-\epsilon}-L}{\lfloor(1+f(n))q\rfloor-L})q^{n-m}\\
%      &=(1-\frac{L\binom{n}{n-m}}{q^{1-\epsilon}}(1-\frac{q^{1-\epsilon}-L}{\lfloor(1+f(n))q\rfloor-L}))|C|\\
%      &\ge \gamma|C|.
% \end{align*}
for  large enough $q\ge q(n,r,L,\gamma)$.
 
 \vspace{0.2cm}{\bf Minimum distance:} By construction, any two codewords of $C$ that agreed on at least $n-m+L(\epsilon+1)$ coordinates were removed,  thus the minimum distance is at least 
 $$m-L(\epsilon+1)+1=\Big\lfloor\frac{L+1}{L}rn \Big\rfloor-(L-1)(\epsilon+1)+1,$$ as needed.
 %the minimum distance
%We show that $d(C')\ge \lfloor\frac{L+1}{L}rn \rfloor-(L-1)(\epsilon+2)+1$.
% Assume otherwise that  there exist two codewords $c,c'\in C'$ that agree on at least  
% $$n-\Big(\Big\lfloor\frac{L+1}{L}rn \Big\rfloor+\epsilon+2-L(\epsilon+2)\Big)=n-m+L(\epsilon+2)$$
% coordinates, and let 
%  $I$ to be any $n-m$ subset   of the these coordinates.    Let $w\in q^{n-m}$ be the restriction of  $c,~c'$ to $I$ and that   this results  in a contradiction because $c,c'$ must have been removed during the construction of $C'$.
 \end{proof}
 
\vspace{0.2cm}
It remains to prove Claim \ref{claim:bad-is-small}

\begin{proof}[Proof of Claim \ref{claim:bad-is-small}]
Let $w$ be a bad vector for $I$, where we assume without loss of generality that $I=[n-m+1,n]$, and assume towards contradiction that there are $L+1$ codewords $c^j$ such that $c^j_{I}=w$. Further, assume   that $c^1,c^2\in C$  $c^1_{[(L\epsilon+1)]}=c^2_{[(L\epsilon+1)]}$.

As before, write $rn=La+b$  with integers  $a$, and $0\leq b<L$. Thus, $m=(L+1)a+b+\epsilon+1$, and since we assumed that $(L-1)(\epsilon+1) \leq \lfloor rn/L\rfloor$, then $L(\epsilon+1)\leq a+\epsilon+1$. Next, partition that set $[2(a+\epsilon+1)-L(\epsilon+1)+1,m]$ into $L-1$ disjoint sets $I_j$ for $j\in [3,L+1]$, each of size at least $a+\epsilon+1$, which is possible since 
\begin{align*}
    &|[2(a+\epsilon+1)-L(\epsilon+1)+1,m]|\\
    &= m-\big(2a-(L-2)(\epsilon+1)+1\big)+1\\
    &= (L+1)a+b+\epsilon+1-\big(2a-(L-2)(\epsilon+1)+1\big)+1\\
    &=(L-1)a+b+(L-1)(\epsilon+1)=(L-1)(a+\epsilon+1)+b
\end{align*}
Next, we show that $C$ is not $(r,L)$ list-decodable. To obtain the desired contradiction we consider the following vector  $y\in [q]^n$ 
$$y_i=\begin{cases}
c^1_i & i \in [a+\epsilon+1]\\
c^2_i & i\in [a+\epsilon+2,2(a+\epsilon+1)-L(\epsilon+1)]\\
c^j_i & i\in I_j \text{~and~} j\in[3,L+1]\\
w_{i-m} & i\in [m+1,n].
\end{cases}$$
%$$y_{[1,a+1]}=c^1_{[1,a+1]},~y_{[a+2,2a+2-L]}=c^2_{[a+2,2a+2-L]}, ~y_{I_j}=c^j_{I_j} \text{~for each~} j\in[3,L+1], \text{~and~} y_{[m+1,n]}=w.$$
 %Clearly $y$ is well-defined as $\{[1,a+1],[a+2,2a+2-L],I_3,\ldots,I_{L+1},[m+1,n]\}$ partitions $[n]$. 
 It is fairly straightforward to check that $c^1,\ldots,c^{L+1}\in B_{rn}(y)$, and we arrive at the desired contradiction.
\end{proof}

\vspace{0.2cm}
{\bf Linear codes:} \textcolor{black}{ By utilizing  the additional structure of linear codes which general codes might not possess,  we can prove that \emph{any} linear code that is list-recoverable  must have  good unique decoding properties, i.e., large minimum distance.  Furthermore, for  $\ell=1, L=1$, this result can also be seen  as 
 a generalization of  the fact that codes that are unique decodable  from relative radius $r$, have a Hamming distance greater than $2rn$ .}
\begin{theorem}\label{thm:linear-distance-recovery}
  For a prime power $q$, integers $1\leq \ell\leq q$, $\ell\leq L<\ell q$ and  $r\in [0,1-\frac{\ell}{L+1})$ with $rn\in\mathbb{N}$,  if $C\subseteq \mathbb{F}_q^n$  is a linear $(r,\ell,L)$ list-recoverable code of dimension at least $2$, then $d(C)> rn+\lfloor\frac{\ell}{L+1-\ell}rn\rfloor$.
  %it has a minimum Hamming distance at least $rn+\lfloor\frac{\ell}{L+1-\ell}rn\rfloor+1$.
\end{theorem}

\textcolor{black}{Note that it is easy to prove that  an $(r,\ell,L)$ list-recoverable linear code has minimum distance  $d(C)>rn$ if $L<q$. Otherwise, all the $q$ multiples of a minimum weight codeword would violate the list-recoverability of the code. %of of is easy, as $B_d(0)$ contains $q$ multiples of the smallest hamming weight codeword. 
Theorem \ref{thm:linear-distance-recovery} improves on  this observation by  utilizing the list-recovery property even further.}
\begin{proof}%[ proof of Theorem \ref{thm:linear-distance-recovery}]
Assume towards contradiction that $d(C)\le rn+\lfloor\frac{\ell}{L+1-\ell}rn\rfloor$, and  let $m:=rn+\lfloor\frac{\ell}{L+1-\ell}rn\rfloor$. Since $C$ is a  linear code, there exists a nonzero codeword  $c\in C$ with weight at most $m$, and assume without loss of generality that $c_{[m+1,n]}=0$. By the dimension of the code, let $c^1\ldots,c^{\ell-1}\in C$ be $\ell$ distinct codewords which are scalar multiples of each other, but  \emph{not} a scalar multiple of $c$, and let $c^\ell=0$. Next, note that any of the  $\ell q>L$ distinct codewords of the form $v=\lambda c+c^j$ for $\lambda\in \mathbb{F}, j\in [\ell]$ satisfy $v_i\in S_i:=\{c^j_i:j\in[\ell]\}$ for $i\in [m+1,n]$, where clearly $|S_i|\leq \ell$. This contradicts  Claim \ref{claim:L-is-max} with $B=S_{m+1}\times \cdots \times S_{n}$.
%The other option is that there is some $c^0\in C$ with $c_{[m+1,n]}\ne 0$. Take any $\ell$ codewords $c^1,..,c^\ell\in <c_0>$ with distinct projections $c^j_{m+1,n}$, and define the lists $S_i:=\{c^j_i:j\in[\ell]\}$ for $i\in [m+1,n]$ and define $B:=\prod\limits_{i\in[m+1,n]}S_i$ thus we have that $<c>+c^j\subseteq B$ for all $j$. But these are $\ell q>L$ different codewords so we also got a contradiction to Claim \ref{claim:L-is-max}. 
\end{proof}

The following proposition is a similar result to that of \cref{thm:linear-distance-recovery} but specialized for list-decodable codes and  without the constraint on the code's dimension to be at least  $2$. We omit its proof, as it is very similar to that of \cref{thm:linear-distance-recovery}.

%result but for list We remark that one can prove a similar result for linear $(r,L)$ list-decodable codes, without the constraint on the code's dimension to be at least  $2$. We omit the proof as it is very similar to that of \cref{thm:linear-distance-recovery}.

%We note that in the case of a stronger condition of $L<q$ the conclusion will hold for all linear codes regardless of the dimension of the code, proof of that is very similar and will be omited. 

\begin{proposition}\label{prop:linear-distance-decoding}
  For a prime power $q$, an integer $1\le L< q$ and  $r\in [0,\frac{L}{L+1})$ with $rn\in\mathbb{N}$, if $C\subseteq \mathbb{F}_q^n$ is a linear $(r,L)$ list-decodable code, then $d(C)> rn+\lfloor\frac{rn}{L}\rfloor$. 
  %it has a minimum Hamming distance at least $\lfloor\frac{L+1}{L}rn\rfloor+1$.
\end{proposition}
 We want to emphasize that the lower bound on the minimum distance given by \cref{prop:linear-distance-decoding} is tight for certain parameters. Indeed, in \cite{shangguan2019combinatorial} it was shown that over sufficiently large finite fields, a positive integer  $rn,~L=2,3$ and $L\mid rn$ there exist an  $[n,n-\frac{L+1}{L}rn]$-RS codes, which are $(r,L)$ list-decodable, and whose minimum distance is $rn+\frac{rn}{L}+1$, which attains the lower bound of \cref{prop:linear-distance-decoding}. 
 
 Proposition \ref{prop:linear-distance-decoding} was also implicitly proved by Roth recently (See  \cite[Theorem 3]{roth2021higherorder}). Specifically,  he shows that an $(r,L)$ list-decodable code of dimension $n-\lfloor\frac{L+1}{L}rn\rfloor$ is an  MDS code. This result  follows directly from  Proposition \ref{prop:linear-distance-decoding}, as shown in the next Corollary, 
% implies that lower bound of the minimum distance given by \cref{prop:linear-distance-decoding} is optimal for those parameters.
%Corollary \ref{cor:linear-distance-decoding} is Theorem \ref{thm:linear-distance-recovery} applied in the case of list-decoding. 
%Recall that  The constructions of in \cite{shangguan2019combinatorial} implies that lower bound of the minimum distance given by \cref{prop:linear-distance-decoding} is optimal for those parameters.
%The following corollary follows directly from \cref{prop:linear-distance-decoding} and 
 which  shows that a linear code that attains with equality the bound on the dimension  given in \cref{prop:linear-bound} has to be a near MDS or an MDS code. 
\begin{corollary}\label{cor:linear-bound-achiver-is-MDS}
  For a prime power $q$, an integer $1\le L< q$ and $r\in [0,\frac{L}{L+1})$ satisfying $rn\in\mathbb{N}$, any $[n,k]$ linear $(r,L)$ list-decodable code whose dimension  attains the bound in \cref{prop:linear-bound} with equality, i.e., $k= n-\lceil\frac{L+1}{L}rn\rceil$, has distance at least $n-k$. Furthermore,  if $L|rn$ the distance is at least $n-k+1$, hence it is an MDS code. 
\end{corollary}

%Note that a linear code $C$ can be this large only when $rn\equiv 0 \pmod L$ by Proposition \ref{prop:linear-bound}. Which is exactly the case considered in the construction of Reed-Solomon codes in \cite{shangguan2019combinatorial}.

\begin{proof}
This is a direct consequence of \cref{prop:linear-distance-decoding}.
%Using Corollary \ref{ccor:linear-distance-decoding} we find that for a $[n,n-\lfloor\frac{L+1}{L}rn\rfloor]_q$ linear code $C$ which is $(r,L)$ list-decodable, thats its minimum Hamming distance must be at least  $$\lfloor\frac{L+1}{L}rn\rfloor+1=n-(n-\lfloor\frac{L+1}{L}rn\rfloor)+1.$$  Hence it is MDS.
\end{proof}

\section*{Acknowledgements}
The  research  of Eitan Goldberg and Itzhak Tamo is partially supported by the European Research Council (ERC grant number 852953) and by the Israel Science Foundation (ISF grant number 1030/15).

The  research  of Chong Shangguan is partially supported by the National Key Research and Development Program of China under Grant No. 2020YFA0712100, the National Natural Science Foundation of China under Grant No. 12101364, the Natural Science Foundation of Shandong Province under Grant
No. ZR2021QA005, and the Qilu Scholar Program of Shandong University. 

{\small\bibliographystyle{alpha}
\bibliography{new,eitannew}}

\end{document}

%% file: setconstrustion.tex
\renewcommand{\arraystretch}{1.4}
\begin{table}[h]
\centering
\begin{tabular}{|c|c|c|c|}
\hline
$S_1$ & $S_2$ & $S_3$ & $S_4$ \\
\hline
\hline
\cellcolor{cyan} $c^1_1$ & \cellcolor{cyan} $c^1_2$ & \cellcolor{cyan} $c^1_3$ & \cellcolor{green} $c^2_4$ \\
\hline
\cellcolor{green} $c^2_1$ & \cellcolor{green} $c^2_2$ & \cellcolor{pink} $c^3_3$ &  \cellcolor{pink} $c^3_4$ \\
\hline
\cellcolor{pink} $c^3_1$ & \cellcolor{orange} $c^4_2$ & \cellcolor{orange} $c^4_3$ & \cellcolor{orange} $c^4_4$ \\
\hline
\end{tabular}
\caption{An example of the sets  $S_i$ for  $rn=1,~L=3,~ \ell=3$, hence $m=4,t=3$. The sets are $S_1=\{c^1_1, c^2_1, c^3_1\},S_2=\{c^1_2, c^2_2, c^4_2\}, S_3=\{c^1_3, c^3_3, c^4_3\}$ and $S_4=\{c^2_4, c^3_4, c^4_4\}$.}
\label{tab:litreview}
\end{table}